\documentclass[letterpaper, 10pt, conference]{ieeeconf}

\IEEEoverridecommandlockouts

\usepackage{amsmath,mathtools}  %
\usepackage{amsfonts,amssymb}

\usepackage{xspace}    %

\usepackage{graphicx}
\graphicspath{{figs/}}

\usepackage{hyperref}

\usepackage{pgfplots}
\usepackage{grffile}
\pgfplotsset{compat=newest}  %
\usetikzlibrary{plotmarks}
\usetikzlibrary{arrows.meta}
\usepgfplotslibrary{patchplots}

\usepackage{tabularx}

\usepackage[ruled]{algorithm}
\usepackage[noend]{algpseudocode}

\makeatletter
\newcommand{\multiline}[1]{%
  \begin{tabularx}{\dimexpr\linewidth-\ALG@thistlm}[t]{@{}X@{}}
    #1
  \end{tabularx}
}
\makeatother

\usepackage[per-mode=symbol]{siunitx}

\usepackage{cite}  %
\usepackage{array}  %
\usepackage{booktabs}           %

\usepackage{cite}

\newtheorem{theorem}{Theorem}

\newtheorem{lemma}[theorem]{Lemma}

\newtheorem{assumption}{Assumption}

\newcounter{remark}
\newenvironment{remark}{%
\par\vspace{3pt}\noindent\refstepcounter{remark}\textbf{Remark~\theremark:}}%
{\par\endtrivlist\unskip}

\newcounter{problem}
{\par\endtrivlist\unskip}

\usepackage[extramath,probability,reviewmark]{truonglatexdefs}

\def\linGP{linGP\xspace}

\newcommand{\GPModel}{\ensuremath{\GGG}}

\newcommand{\GPmean}{\ensuremath{\mu}}

\usepackage{stfloats}%

\title{Distributed Experiment Design and Control for Multi-agent Systems with Gaussian Processes}

\author{Viet-Anh Le and Truong X. Nghiem\\
  School of Informatics, Computing, and Cyber Systems\\
  Northern Arizona University\\
  \{vl385,truong.nghiem\}@nau.edu%
}

\begin{document}

\maketitle
\thispagestyle{empty}
\pagestyle{empty}

\begin{abstract}
  This paper focuses on distributed learning-based control of decentralized multi-agent systems where the agents' dynamics are modeled by Gaussian Processes (GPs).
  Two fundamental problems are considered: the optimal design of experiment for concurrent learning of the agents' GP models, and the distributed coordination given the learned models.
  Using a Distributed Model Predictive Control (DMPC) approach, the two problems are formulated as distributed optimization problems, where each agent's sub-problem includes both local and shared objectives and constraints.
  To solve the resulting complex and non-convex DMPC problems efficiently, we develop an algorithm called Alternating Direction Method of Multipliers with Convexification (ADMM-C) that combines a distributed ADMM algorithm and a Sequential Convexification method.
  The computational efficiency of our proposed method comes from the facts that the computation for solving the DMPC problem is distributed to all agents and that efficient convex optimization solvers are used at the agents for solving the convexified sub-problems.
  We also prove that, under some technical assumptions, the ADMM-C algorithm converges to a stationary point of the penalized optimization problem.
  The effectiveness of our approach is demonstrated in numerical simulations of a multi-vehicle formation control example.
\end{abstract}

\section{Inroduction}
\label{sec:intro}

Multi-agent control systems have been studied extensively in recent decades due to their increasing number of applications such as building energy networks, smart grids, robotic swarms, and wireless sensor networks.
The majority of control methods designed for single systems cannot be easily extended to multi-agent control systems due to additional challenges such as the combination of global and local tasks, limited communication and computation capabilities, and privacy requirements that limit information sharing between agents \cite{cao2012overview}.
While the centralized approach where a coordinator is available to coordinate and manipulate all agents, either with distributed computation or not, facilitates the communication and information sharing between agents, it does not scale reasonably with a large number of agents due to physical constraints such as short communication ranges, or the limited number of connections to the coordinator. 
For this reason, recent studies have been focused on decentralized multi-agent control systems, in which the coordinator is eliminated and each agent in the network can communicate and collaborate with a few other agents, called neighbors, to achieve the desired control objectives.

Among various control methods for single dynamical systems, Model Predictive Control (MPC) is an advanced control technique that has been widely adapted to multi-agent systems due to its flexibility and efficiency in handling multiple control objectives and constraints. 
The extension of MPC for multi-agent systems is widely known as Distributed MPC (DMPC) \cite{negenborn2014distributed}.
To solve a DMPC problem in a distributed manner, distributed optimization algorithms are commonly used. 
In \cite{raffard2004distributed}, the authors consider an optimization control problem of flight formation and develop an algorithm to solve it based on dual decomposition techniques. 
In \cite{summers2012distributed}, the Alternating Direction Method of Multipliers (ADMM) %
was utilized for solving a DMPC problem. 
In \cite{conte2012computational}, the authors provided a computational study on the performance of two distributed optimization algorithms, the dual decomposition based on fast gradient updates (DDFG) and the ADMM, for DMPC problems.
Some other distributed optimization algorithms used for DMPC are fast alternating minimization algorithm (FAMA) and inexact FAMA \cite{pu2014inexact}, inexact Proximal Gradient Method and its accelerated variant \cite{pu2015quantization}. 
In terms of applications, DMPC has been applied for numerous practical multi-agent systems such as robotic swarms \cite{van2017distributed,luis2020online}, and building energy networks \cite{ma2011distributed,hou2017distributed}.

In the above works, the dynamics of all agents are assumed to be available and sufficiently precise.
However, for many complex dynamical systems, accurately modeling the system dynamics based on physics is often not straightforward due to the existence of uncertainties and ignored dynamical parts. 
This challenge motivated us to develop learning-based DMPC for multi-agent systems in our previous paper \cite{le2020gaussian}, where Gaussian Processes (GPs) \cite{williams2006gaussian} were employed to learn the agent non-linear dynamics resulting in a GP-based DMPC (GP-DMPC) problem.
To solve the GP-DMPC problem, a distributed optimization algorithm, called \linGP-SCP-ADMM, was developed to solve the GP-DMPC problem efficiently. %
However, in \cite{le2020gaussian}, we assumed that the GP dynamics of all agents are identical and available, which may not hold in practical applications since each agent has its own dynamics or system parameters.
The problem pertaining to how to obtain training datasets for all agents \emph{in one experiment} was thus not addressed. 
Moreover, the convergence properties of the \linGP-SCP-ADMM algorithm was not analyzed in our work.

Therefore, in this paper, we formulate a GP-DMPC problem that covers two fundamental problems of learning-based control for decentralized multi-agent systems, namely experiment design and coordination problems. 
In the experiment design problem, we utilize the receding horizon active learning approach \cite{le2021receding} with exact conditional differential entropy to include individual learning objectives into the DMPC problem.  
To solve the non-convex and complex GP-DMPC problem, we develop a new algorithm called ADMM with Convexification (ADMM-C) that combine the distributed ADMM optimization method and Sequential Convexification Programming (SCP) technique \cite{mao2018successive}.
Note that the ADMM-C is different from the \linGP-SCP-ADMM presented in our previous work \cite{le2020gaussian}.
In \linGP-SCP-ADMM, at each iteration, we used the \textit{linearized Gaussian Process} (\linGP) \cite{nghiem2019linearized} and SCP method to form a convex GP-DMPC subproblem that can be solved by convex distributed ADMM algorithm \cite{boyd2011distributed}, but this method is not applicable for the problem considered in this paper, where the active learning objective is included.
Meanwhile, the ADMM-C in this paper is a variant of the ADMM algorithm for non-convex and non-smooth optimization \cite{wang2019global} where the convexification technique is used to solve the non-convex local subproblems at each ADMM iteration.
In addition, the \linGP-SCP-ADMM algorithm was dedicated for the multi-agent system with a coordinator, whereas ADMM-C in this paper is designed for decentralized systems.   
Under some technical assumptions, we prove that the ADMM-C algorithm converges to a stationary point of the penalized GP-DMPC problem.
The effectiveness of our algorithm is demonstrated in a simulation case study of experiment design and formation control problem for a multi-vehicle system.

The remainder of this paper is organized as follows. 
The GP-DMPC formulation for distributed experiment design and coordination of a multi-agent system is introduced in Section~\ref{sec:problem}. 
Our proposed ADMM-C algorithm is presented in Section~\ref{sec:algorithm} and the simulation results are reported and discussed in Section~\ref{sec:simulation}. 
Finally, Section~\ref{sec:conclusion} concludes the paper with a summary and some future directions.

\section{Problem Formulation}
\label{sec:problem}

This section introduces a Gaussian Process-based Distributed Model Predictive Control (GP-DMPC) formulation for distributed experiment design and control problems of a multi-agent system, in which Gaussian Processes (GPs) are employed to represent the agent dynamics.
Our problem formulation covers two fundamental problems: (1) the \textit{multi-agent experiment design problem} based on active learning where each agent explores the state-space to collect informative data for system identification while guaranteeing certain cooperative objectives and constraints with other agents, and (2) the \textit{multi-agent coordination problem} in which the agents cooperate to achieve both local and shared objectives using the obtained GP dynamic models.

Consider a decentralized multi-agent control system involving $M$ dynamical agents.
We assume bidirectional communication between the agents, \ie if agent $i$ can communicate with agent $j$, then agent $j$ can communicate with agent $i$. 
Consequently, the communication between agents in this network is described by an undirected graph $G = (\VVV, \EEE)$ where $\VVV = \{ 1, 2, \dots, M\}$ is the vertex set representing the agents, and $\EEE$ is the edge set defining the connections between pairs of agents, \ie $(i, j) \in \EEE$ means that agents $i$ and $j$ are neighbors.
Moreover, we define $\NNN_i = \{ j | (i,j) \in \EEE\}$ as the set of agent $i$'s neighbors (we assume that $i \in \NNN_i$) and the number of elements in the set $\NNN_i$ is denoted by $| \NNN_i |$.

For every agent $i$, we define its vector of control inputs as $\mathbf{u}_{i} \in \RR^{n_{u,i}}$, 
its vector of GP output variables as $\mathbf{y}_i \in \RR^{n_{y,i}}$,
and its vector of non-GP variables as $\mathbf{z}_i \in \RR^{n_{z,i}}$.
For any variable $\square_i$ of agent $i$, where $\square$ is $\mathbf{y}$, $\mathbf{z}$, or $\mathbf{u}$, let $\square_{i,k}$ denote its value at time step $k$.
The GP dynamics of agent $i$ express $\mathbf{y}_i$ as $\mathbf{y}_{i,k} \sim \GPModel_i(\mathbf{x}_{i,k}; \mathrm{m}_i, \mathrm{k}_i)$, where $\GPModel_i(\cdot; \mathrm{m}_i, \mathrm{k}_i)$ is a GP with mean function $\mathrm{m}_i$ and covariance function $\mathrm{k}_i$.
The input vector $\mathbf{x}_{i,k}$ of the GP is formed from current and past values of the control inputs $\mathbf{u}_{i,\tau}$ and non-GP states $\mathbf{z}_{i,\tau}$, for $\tau \leq k$, as well as from past GP outputs $\mathbf{y}_{i,\tau}$, for $\tau < k$.
Given an input $\mathbf{x}_{i,k}$, let $\mathbf{\bar{y}}_{i,k} = \GPmean_i(\mathbf{x}_{i,k})$ denote the predicted mean of the GP model $\GPModel_i(\cdot; \mathrm{m}_i, \mathrm{k}_i)$ at $\mathbf{x}_{i,k}$.
Note that in this paper, we only utilize the GP means without uncertainty propagation to represent the predicted values of the nonlinear dynamics.
More details on GP regression for dynamics and control can be found in \cite{jain2018learning,beckers2019stable}.

Let $H > 0$ be the length of the MPC horizon, $t$ be the current time step and $\III_t = \{ t, \dots, t+H-1 \}$ be the set of all time steps in the MPC horizon at time step $t$.
Denote $\bar{\YYY}_{i,t} = \{ \mathbf{\bar{y}}_{i,k} | k \in \III_t \}$, $\ZZZ_{i,t} = \{ \mathbf{z}_{i,k} | k \in \III_t \}$, $\UUU_{i,t} = \{ \mathbf{u}_{i,k} | k \in \III_t \}$, and $\XXX_{i,t} = \{ \mathbf{x}_{i,k} | k \in \III_t \}$ as the sets collecting the predicted GP output means, the non-GP states, the control inputs, and the GP inputs of agent $i$ over the MPC horizon.
For each agent $i$, we define the concatenated vectors $\square_{\NNN_i,k}$ of local variables $\square_{j,k}$ of all agents $j \in \NNN_i$, where $\square$ is $\mathbf{\bar{y}}$, $\mathbf{z}$, and $\mathbf{u}$.
Correspondingly, the collections $\bar{\YYY}_{\NNN_i,t}$, $\ZZZ_{\NNN_i,t}$, and $\UUU_{\NNN_i,t}$ of $\mathbf{\bar{y}}_{\NNN_i,k}$, $\mathbf{z}_{\NNN_i,k}$, and $\mathbf{u}_{\NNN_i,k}$ over the MPC horizon are defined.
We also utilize the notation $[X]$ to denote the vector concatenation of all vectors in a set $X$ (\eg $[\XXX_i] = [\mathbf{x}_{i,k}^T]_{k \in \III_t}^T$).

To facilitate the problem formulation, we present the $H$ mean equations of GP dynamics in the current control horizon 
as $\mu_{i,l} (\bar{\YYY}_{i}, \XXX_i) = 0$ for $l \in \III_t$.
Moreover, for each agent $i$, let's define $\HHH_i(\XXX_i)$ as an individual objective for active learning, 
$J_i (\bar{\YYY}_{i}, \UUU_{i}, \ZZZ_{i})$ as a local control objective, 
$J_{\NNN_i} (\bar{\YYY}_{\NNN_i}, \UUU_{\NNN_i}, \ZZZ_{\NNN_i})$ as a shared objective,  
$g_{i,l}(\bar{\YYY}_i, \UUU_i, \ZZZ_i) \le 0,\; \forall l \in \III_{\text{ieq}}$ and $h_{i,l}(\bar{\YYY}_i, \UUU_i, \ZZZ_i) = 0,\; \forall l \in \III_{\text{eq}}$ as inequality and equality constraints where $\III_{\text{ieq}}$ and $\III_{\text{eq}}$ are sets of inequality and equality constraint indices, respectively.
As a result, the GP-DMPC problem for distributed experiment design and coordination of the multi-agent system is formulated as follows, where the current time-step subscript $t$ is omitted for brevity. 
\vspace{-5pt}
\begin{equation}
\label{eq:gpdmpc}
\begin{split}
  &
  \begin{multlined}
  \underset{\{\UUU_{i}, \ZZZ_{i}\}_{i \in \VVV}}{\text{minimize}} \quad \sum_{i=1}^M J_i(\bar{\YYY}_{i}, \UUU_{i}, \ZZZ_{i}) - \gamma \HHH_i(\XXX_i) \\
  + J_{\NNN_i} (\bar{\YYY}_{\NNN_i}, \UUU_{\NNN_i}, \ZZZ_{\NNN_i})
  \end{multlined} \\
  & \text{subject to} \\
  & \quad \mu_{i,l} (\bar{\YYY}_{i}, \XXX_i) = 0 ,\; \forall i \in \VVV, \; \forall l \in \III_{t} \\
  & \quad h_{i,l}( \bar{\YYY}_{i}, \UUU_{i}, \ZZZ_{i} ) = 0 ,\; \forall i \in \VVV, \; \forall l \in \III_{i,eq} \\
  & \quad g_{i,l}( \bar{\YYY}_{i}, \UUU_{i}, \ZZZ_{i} ) \le 0 ,\; \forall i \in \VVV, \; \forall l \in \III_{i,ieq} 
\end{split}
\end{equation} 
where $\gamma$ is a positive constant representing a tradeoff between learning and control objectives and note that $\gamma > 0$ in the experiment design problem while $\gamma = 0$ in the cooperative control problem. 
We use the exact conditional differential entropy of multi-step ahead GP predictions \cite{le2021receding} to represent the individual active learning objective function instead of an upper bound as in \cite{buisson2020actively}, \ie $\HHH_i(\XXX_i) = \operatorname{{log\,det}} \big( \mathbf{\Sigma}_{\GPModel_f} ([\XXX_i]) \big)$ where $\mathbf{\Sigma}_{\GPModel_f}$ is a $H\times H$ posterior covariance matrix of GP joint predictions at the $H$ inputs in $\XXX_i$.
It was shown in \cite{le2021receding} that the optimal experiment using the active learning approach can significantly improve data quality for model learning in comparison with randomized experiments or with using only historical data.
For more details on the active learning technique for dynamical GPs using multi-step ahead prediction approach for a single dynamical system, the readers are referred to \cite{le2021receding,buisson2020actively,capone2020localized}.

\begin{remark}
As discussed in \cite{wang2019global}, any feasible constraint set on the shared variables can be treated as an indicator function and included in the shared objective function $ J_{\NNN_i} (\bar{\YYY}_{\NNN_i}, \UUU_{\NNN_i}, \ZZZ_{\NNN_i})$. 
Therefore, to simplify the problem formulation, we do not include constraints on the shared variables.
In contrast, we present constraints on the local variables as equality and inequality constraints since they are encoded into the local objective function by corresponding penalized functions in our method. 
\end{remark}

We make the following technical assumptions related to the problem \eqref{eq:gpdmpc}.

\begin{assumption}
\label{A-feasibility}
The original problem \eqref{eq:gpdmpc} is feasible.
\end{assumption}

\begin{assumption}
\label{A-convexity}
For all $i \in \VVV$, $J_i$ and $g_{i,l}$ are convex, continuous and Lipschitz differentiable, $h_{i,l}$ are affine, continuous and Lipschitz differentiable in the variables $(\bar{\YYY}_{i}, \UUU_{i}, \ZZZ_{i})$. 
\end{assumption}

\begin{assumption}
\label{A5-admm}
For all $i \in \VVV$, $J_{\NNN_i}$ are convex, continuous and Lipschitz differentiable in the variables $(\bar{\YYY}_{\NNN_i}, \UUU_{\NNN_i}, \ZZZ_{\NNN_i})$. 
\end{assumption}

From Assumptions~\ref{A-convexity} and \ref{A5-admm}, it can be seen that the non-convexity of the problem \eqref{eq:gpdmpc} only results from the GP dynamics and the active learning objectives. 

\begin{assumption}
\label{A3.15-scp}
For all $i \in \VVV$, $\mu_{i,l} (\bar{\YYY}_{i}, \XXX_i), \; \forall l \in \III_{t}$, and $\HHH_i(\XXX_i)$ are continuously Lipschitz differentiable in the variables $(\bar{\YYY}_{\NNN_i}, \XXX_{\NNN_i})$.
\end{assumption}

To overcome the complexity and non-convexity of the original problem \eqref{eq:gpdmpc}, we will convexify the non-convex terms by using their first-order approximations.
However, to avoid the \textit{artificial infeasibility} \cite{mao2018successive} of the problem due to these approximations, the inequality and equality constraints in \eqref{eq:gpdmpc} are encoded into the objective function by the exact penalty functions \cite{mao2018successive} leading to the following penalized optimization problem
\vspace{-5pt}
\begin{equation}
\label{eq:gpdmpc-pen}
\begin{split}
  &
  \begin{multlined}
  \underset{\{\UUU_{i}, \ZZZ_{i}\}_{i \in \VVV}}{\text{minimize}} \quad \sum_{i=1}^M J_i(\bar{\YYY}_{i}, \UUU_{i}, \ZZZ_{i}) - \gamma \HHH_i(\XXX_i) \\
  + \tau_i \Big( \sum_{l \in \III_{t}} \left | \mu_{i,l} (\bar{\YYY}_{i}, \XXX_i) \right | 
  + \sum_{l \in \III_{i,eq}} \left | h_{i,l}( \bar{\YYY}_{i}, \UUU_{i}, \ZZZ_{i} ) \right |  \Big) \\
  + \lambda_i  \sum_{l \in \III_{i,ieq}} \operatorname{max} \big( 0, g_{i,l}( \bar{\YYY}_{i}, \UUU_{i}, \ZZZ_{i}) \big) \\
  + J_{\NNN_i} (\bar{\YYY}_{\NNN_i}, \UUU_{\NNN_i}, \ZZZ_{\NNN_i}) 
  \end{multlined} 
\end{split}
\end{equation} 
where $\tau_i$ and $\lambda_i$, $\forall i \in \VVV$ are large positive penalty weighs.

\begin{assumption}
\label{A-coercivity}
The objective function of the penalized problem \eqref{eq:gpdmpc-pen} is \textit{coercive} \cite{wang2019global}.
\end{assumption}

\section{ADMM with Convexification for GP-DMPC}
\label{sec:algorithm}

In this section, we propose a distributed optimization algorithm called Alternating Direction Method of Multipliers with Convexification (ADMM-C) that is based on the ADMM algorithm \cite{boyd2011distributed} and SCP technique \cite{mao2018successive} for solving the complex and non-convex problem \eqref{eq:gpdmpc}.
The ADMM algorithm was designed for solving convex large-scale optimization problems in a distributed manner \cite{boyd2011distributed}.
For non-convex and non-smooth optimization problems like the problem \eqref{eq:gpdmpc-pen}, the ADMM for non-convex non-smooth optimization \cite{wang2019global} was developed.
However, the algorithm design in \cite{wang2019global} requires all non-convex optimization subproblems to be solved exacly at each iteration, which might restrict its usage in real-time applications.
Moreover, the complexity of the log determinant of the GP covariance matrix in \eqref{eq:gpdmpc-pen} makes the non-convex local subproblems computationally intractable for nonlinear programming solvers.
Therfore, ADMM-C is developed in this section by sequentially convexifying the non-convex local subproblems at each iteration.

To facilitate the algorithm design, we rewrite the penalized GP-DMPC problem \eqref{eq:gpdmpc-pen} in the following simplified form
\begin{equation}
  \label{eq:dop-def}
  \underset{\{\mathbf{x}_i\}_{i \in \VVV}}{\text{minimize}}\; \textstyle\sum_{i=1}^{M} f_i(\mathbf{x}_i) + f_{\NNN_i} (\mathbf{x}_{\NNN_i}) 
\end{equation}
in which $\mathbf{x}_i \in \RR^{n_i}$ is the vector collecting the local variables of agent $i$, and $\mathbf{x}_{\NNN_i} \in \RR^{n_{\NNN_i}}$, where $n_{\NNN_i} = \sum_{j \in \NNN_i} n_j$, are the shared vector concatenating the local variables of all the neighbors of agent $i$, \ie $\mathbf{x}_{\NNN_i} = [\mathbf{x}_{j}^T]_{j \in \NNN_i}^T$.
Moreover, let $F_{ij}$ denote the matrix of transformation between the local variables of agent $i$ and the vector of shared variables of agent $j$ for each $(i, j) \in \EEE$, \ie $\mathbf{x}_i = F_{ij} \mathbf{x}_{\NNN_j}$.
The local objective functions $f_i(\mathbf{x}_i)$ and shared objective functions $f_{\NNN_i} (\mathbf{x}_{\NNN_i})$ are respectively defined by
\begin{equation*}
\begin{multlined}
f_i(\mathbf{x}_i) := J_i(\bar{\YYY}_{i}, \UUU_{i}, \ZZZ_{i}) - \gamma \HHH_i(\XXX_i) \\
+ \tau_i \Big( \sum_{l \in \III_{t}} \left | \mu_{i,l} (\bar{\YYY}_{i}, \XXX_i) \right | 
+ \sum_{l \in \III_{i,eq}} \left | h_{i,l}( \bar{\YYY}_{i}, \UUU_{i}, \ZZZ_{i} ) \right |  \Big) \\
+ \lambda_i  \sum_{l \in \III_{i,ieq}} \operatorname{max} \big( 0, g_{i,l}( \bar{\YYY}_{i}, \UUU_{i}, \ZZZ_{i}) \big)
\end{multlined}
\end{equation*}
and
\begin{equation*}
f_{\NNN_i} (\mathbf{x}_{\NNN_i}) := J_{\NNN_i} (\bar{\YYY}_{\NNN_i}, \UUU_{\NNN_i}, \ZZZ_{\NNN_i})
\end{equation*}

The ADMM-C algorithm solves the problem \eqref{eq:dop-def} in the following consensus form
\begin{equation}
  \label{eq:admm-prob}
  \begin{split}
    \underset{\{\mathbf{x}_i, \mathbf{z}_{\NNN_i}\}_{i \in \VVV}}{\text{minimize}}\; & \textstyle\sum_{i=1}^{M} f_i(\mathbf{x}_i) + f_{\NNN_i}(\mathbf{z}_{\NNN_i}) \\
    \text{subject to}\; & \mathbf{x}_{\NNN_i} = \mathbf{z}_{\NNN_i},\; \forall i \in \VVV \text.
  \end{split}
\end{equation}
with a copy $\mathbf{z}_{\NNN_i}$ of $\mathbf{x}_{\NNN_i}$. The augmented Lagrangian for problem \eqref{eq:admm-prob} is
\begin{equation}
  \label{eq:admm-lagrang1}
  L_\rho( \{ \mathbf{x}_{i}, \mathbf{z}_{\NNN_i}, \mathbf{y}_{\NNN_i} \}_{i \in \VVV} ) 
  = \sum_{i \in \VVV} L_{\rho,i} (  \mathbf{x}_{\NNN_i}, \mathbf{z}_{\NNN_i}, \mathbf{y}_{\NNN_i} )
\end{equation}
where 
\begin{equation*}
\begin{multlined}
L_{\rho,i} (  \mathbf{x}_{\NNN_i}, \mathbf{z}_{\NNN_i}, \mathbf{y}_{\NNN_i} ) 
=  f_i(F_{ii} \mathbf{x}_{\NNN_i}) + f_{\NNN_i}(\mathbf{z}_{\NNN_i}) \\
+ \mathbf{y}_{\NNN_i}^T (\mathbf{x}_{\NNN_i} - \mathbf{z}_{\NNN_i}) 
+ \frac{\rho}{2} \norm{\mathbf{x}_{\NNN_i} - \mathbf{z}_{\NNN_i}}_2^2
\end{multlined}
\end{equation*}
and $\mathbf{y}_{\NNN_i}$, $i \in \VVV$, is the associated dual variables. 
Note that the $\mathbf{x}$, $\mathbf{y}$ and $\mathbf{z}$ notations in this section are different from those in Section~\ref{sec:problem} which were used in the system dynamics.

Since the concensus constraint $\mathbf{x}_{\NNN_i} = \mathbf{z}_{\NNN_i}$, $\forall i \in \VVV$, in \eqref{eq:admm-prob} can be replaced by $\mathbf{x}_i = F_{ij} \mathbf{z}_{N_j}$ and $\mathbf{x}_j = F_{ji} \mathbf{z}_{N_i}$, $\forall (i, j) \in \EEE$, the Lagrangian \eqref{eq:admm-lagrang1} is equivalent to the following Lagrangian 
\begin{equation}
\label{eq:admm-lagrang2}
L_\rho( \{ \mathbf{x}_i, \mathbf{z}_{\NNN_i}, \mathbf{y}_{i} \}_{i \in \VVV} ) 
= \sum_{i \in \VVV} \bar{L}_{\rho,i} (  \mathbf{x}_{i}, \mathbf{z}_{\NNN_i}, \mathbf{y}_{\NNN_i} )
\end{equation}
where 
\begin{equation*}
\begin{multlined}
\bar{L}_{\rho,i} (  \mathbf{x}_{i}, \mathbf{z}_{\NNN_i}, \mathbf{y}_{\NNN_i} ) 
= f_i(\mathbf{x}_i) + f_{\NNN_i}(\mathbf{z}_{\NNN_i}) \\
+ \sum_{j \in \NNN_i} \Big( (F_{ij} \mathbf{y}_{\NNN_j})^T (\mathbf{x}_i - F_{ij} \mathbf{z}_{N_j}) 
+ \frac{\rho}{2} \norm{\mathbf{x}_i - F_{ij} \mathbf{z}_{N_j}}_2^2 \Big) 
\end{multlined}
\end{equation*}

Using the augmented Lagrangians \eqref{eq:admm-lagrang1} and \eqref{eq:admm-lagrang2}, the ADMM algorithm\cite{boyd2011distributed} consists of three following steps:
\begin{subequations}
\label{eq:admm-step}
\begin{align}
\mathbf{x}_{i}^{(k+1)} &= \underset{\mathbf{x}_{i}}{\text{argmin}}\; \bar{L}_{\rho, i} ( \mathbf{x}_i, \mathbf{z}_{\NNN_i}^{(k)}, \mathbf{y}_{\NNN_i}^{(k)} ), \; \forall i \in \VVV  \label{eq:admm-x}\\
\mathbf{z}_{\NNN_i}^{(k+1)} &= \underset{\mathbf{z}_{\NNN_i}}{\text{argmin}}\; L_{\rho, i} ( \mathbf{x}_{\NNN_i}^{(k+1)}, \mathbf{z}_{\NNN_i}, \mathbf{y}_{\NNN_i}^{(k)} ),  \; \forall i \in \VVV \label{eq:admm-z}\\
\mathbf{y}_{\NNN_i}^{(k+1)} &= \mathbf{y}_{\NNN_i}^{(k)} + \rho (\mathbf{x}_{\NNN_i}^{(k+1)} - \mathbf{z}_{\NNN_i}^{(k+1)}), \; \forall i \in \VVV
\end{align}
\end{subequations}

The $x$-minimization step \eqref{eq:admm-x} is equivalent to the following proximal operator
\begin{equation}
\label{eq:admm-x-prox}
\mathbf{x}_{i}^{(k+1)} 
= \textbf{prox}_{\frac{1}{\rho |\NNN_i|} f_{i}} \Bigg(\frac{\sum\limits_{j\in \NNN_i} F_{ij} (\mathbf{z}_{N_j}^{(k)}- \mathbf{y}_{\NNN_i}^{(k)}/\rho)}{|\NNN_i|}  \Bigg)
\end{equation}
while the $z$-minimization step \eqref{eq:admm-z} can be rewritten in the form of the proximal minimization problem
\begin{equation}
\label{eq:admm-z-prox}
\mathbf{z}_{\NNN_i}^{(k+1)} 
= \textbf{prox}_{\frac{1}{\rho} f_{\NNN_i}}(\mathbf{x}_{\NNN_i}^{(k+1)} + \frac{\mathbf{y}_{\NNN_i}^{(k)}}{\rho})
\end{equation}

Note that under Assumptions~\ref{A-convexity} and \ref{A5-admm}, the subproblem \eqref{eq:admm-z-prox} is convex, while the subproblem \eqref{eq:admm-x-prox} is non-convex due to the GP dynamics and the active learning objectives.
In the ADMM-C optimization algorithm, we propose to solve the non-convex subproblem \eqref{eq:admm-x} in each iteration by convexification method.
In particular, instead of solving the non-convex subproblem \eqref{eq:admm-x}, $\mathbf{x}_{i}^{(k+1)}$ is determined by
\begin{equation}
\label{eq:admm-x-scp}
\mathbf{x}_{i}^{(k+1)} = \mathbf{x}_{i}^{(k)} +  \Delta \mathbf{x}_{i}^{(k+1)} 
\end{equation}
where 
\vspace{-5pt}
\begin{equation*}
\label{eq:admm-deltax-scp}
\Delta \mathbf{x}_{i}^{(k+1)}
= \underset{\norm{\Delta \textbf{x}_{i}} \leq r_i}{\textbf{prox}_{\frac{1}{\rho |\NNN_i|} \tilde{f}_{i}}} \Bigg( \frac{\sum\limits_{j\in \NNN_i} F_{ij} \mathbf{z}_{N_j}^{(k)}}{|\NNN_i|} 
- \frac{\mathbf{y}_{i}^{(k)}}{\rho} 
- \mathbf{x}_i^{(k)} \Bigg) 
\end{equation*}
where $\tilde{f}_{i}$ is the approximated function of $f_{i}$ in a small trust region with radius $r_i$ around the nominal solution $\mathbf{x}_{i}^{(k)}$, as shown in \eqref{eq:admm-tilde-f} on the next page,
where $\tilde{\HHH}_i(\Delta \XXX_i)$ and $\tilde{\mu}_{i,l} (\Delta\YYY_i, \Delta\XXX_i)$ are first-order approximations of $\HHH_i(\XXX_i)$ and $\mu_{i,l} (\bar{\YYY}_{i}, \XXX_i)$ around nominal values $\XXX_{i}^{\star}$ and $\YYY_{i}^{\star}$.
The details on the first-order approximations of the log determinant of GP covariance matrix and the GP predicted mean can be found in \cite{le2021receding}.
\begin{figure*}[t!]
\begin{equation}
\label{eq:admm-tilde-f}
\begin{multlined}
\tilde{f}_{i} (\Delta \mathbf{x}_{i}) := J_i(\YYY_{i}^{\star} + \Delta\YYY_i, \UUU_{i}^{\star} + \Delta\UUU_{i}, \ZZZ_{i}^{\star} + \Delta\ZZZ_i) - \gamma \tilde{H}_i(\Delta \XXX_i)
+ \tau_i \Big( \sum_{l \in \III_{t}} \left | \tilde{\mu}_{i,l} (\Delta\YYY_i, \Delta\XXX_i) \right | \\
+ \sum_{l \in \III_{i,eq}} \left | h_{i,l}(\YYY_{i}^{\star} + \Delta\YYY_i, \UUU_{i}^{\star} + \Delta\UUU_{i}, \ZZZ_{i}^{\star} + \Delta\ZZZ_i) \right |  \Big)
+ \lambda_i  \sum_{l \in \III_{i,ieq}} \operatorname{max} \big( 0, g_{i,l}(\YYY_{i}^{\star} + \Delta\YYY_i, \UUU_{i}^{\star} + \Delta\UUU_{i}, \ZZZ_{i}^{\star} + \Delta\ZZZ_i) \big)
\end{multlined} 
\end{equation}
\hrule
\vspace{-18pt}
\end{figure*}

Under the Assumption~\ref{A-convexity}, the problem \eqref{eq:admm-x-scp} is convex and thus can be solved by a convex solver.
The obtained solution is then decided to be accepted or rejected, and the trust-region radius is adapted as elaborated in \cite{nghiem2019linearized}.
Therefore, we obtain the ADMM-C algorithm as presented in Algorithm~\ref{alg:ADMM-C}.

\begin{remark}
\label{rmk-rules}
Each agent computes the actual local cost reduction $\delta_i^{(k+1)} = f_i(\mathbf{x}_i^{(k)}) - f_i(\mathbf{x}_i^{(k+1)})$ and the predicted local cost reduction, \ie $\tilde{\delta}_i^{(k+1)} = f_i(\mathbf{x}_i^{(k)}) - \tilde{f}_i(\mathbf{\Delta x}_i^{(k+1)})$, and compare the ratio $\delta_i^{(k+1)}/\tilde{\delta}_i^{(k+1)}$ with some predefined thresholds $0 < \epsilon_0 < \epsilon_1 < \epsilon_2 < 1$ to adjust the trust region $r_i$ by $\beta_{\mathrm{fail}} < 1$ and $\beta_{\mathrm{succ}} > 1$ according to the adjustment rule in \cite{nghiem2019linearized}.    
\end{remark}

\begin{remark}
Since there is no central coordinator to supervise the practical convergence of the problem, the algorithm is terminated when a predefined number of iterations is reached \cite{van2017distributed}.
\end{remark}

The convergence analysis of the proposed ADMM-C algorithm will be given in the Appendix~\ref{sec:convergence} .

\begin{algorithm}[t]
  \caption{ADMM-C algorithm}
  \label{alg:ADMM-C}
  \begin{algorithmic}[1]
    \Require $\mathbf{y}_{\NNN_i}^{(0)}$, $\mathbf{z}_{\NNN_i}^{(0)}$, $\rho > 0$, $k_{\mathrm{max}} > 0$.
    \For {$k = 1, \dots, k_{\mathrm{max}}$}
    \State Agent $i$ sends $F_{ij} (\mathbf{z}_{\NNN_i}^{(k)} - \mathbf{y}_{\NNN_i}^{(k)}/\rho)$ to its neighbors $j \in \NNN_i$.
    \State Agent $i$ computes $\mathbf{x}_i^{(k+1)}$ by \eqref{eq:admm-x-scp}. %
    \State Agent $i$ accepts or rejects the obtained solution of \eqref{eq:admm-x-scp}, and updates trust-region radius by the Remark~\ref{rmk-rules}.
    \State Agent $i$ sends $\mathbf{x}_i^{(k+1)}$ to its neighbors $j \in \NNN_i$.
    \State Agent $i$ receives all $\mathbf{x}_j^{(k+1)}$ from its neighbors $j \in \NNN_i$ and forms $\mathbf{x}_{\NNN_i}^{(k+1)}$.
    \State Agent $i$ computes $\mathbf{z}_{\NNN_i}^{(k+1)}$ by \eqref{eq:admm-z-prox}
    \State Agent $i$ updates $\mathbf{y}_{\NNN_i}^{(k+1)} = \mathbf{y}_{\NNN_i}^{(k)} + \rho(\mathbf{x}_{\NNN_i}^{(k+1)} - \mathbf{z}_{\NNN_i}^{(k+1)})$
    \EndFor
     \State \textbf{return $\mathbf{x}_i^{(k_{\mathrm{max}})}$ and $\mathbf{x}_{\NNN_i}^{(k_{\mathrm{max}})}$} 
  \end{algorithmic}
\end{algorithm} \setlength{\textfloatsep}{0.05cm}

\section{Simulation}
\label{sec:simulation}

In this section, we utilize the multi-agent formation control example presented in \cite{raffard2004distributed} as an illustrative example for the distributed experiment design and control problem and the ADMM-C algorithm.

\subsection{Multi-vehicle formation control example}
\label{sec:simulation:example}

We consider a group of $M = 2N+1$ vehicles where the communication between the agents is specified by the edge set $\EEE = \{(i,j)\; | \; |i-j| = 1 \}$, \ie two agents are neighbors if and only if their indices are consecutive.
The dynamics of each vehicle $i$ is described by the following continuous-time kinematic bicycle model \cite{kong_kinematic_2015}
\begin{equation}
\begin{alignedat}{2}
  \dot{x}_i &= v_i \cos(\theta_i + \beta_i), \quad
  && \dot{y}_i = v_i \sin(\theta_i + \beta_i), \\
  \dot{\theta}_i &= \frac{v_i}{l_{r,i}} \sin (\beta_i), \quad
  && \dot{v}_i = a_i .
\end{alignedat}
\vspace{-10pt}
\end{equation}
where $\beta_i = \tan^{-1} \left( \frac{l_{r,i}}{l_{f,i} + l_{r,i}} \tan (\alpha_i)\right)$ is the angle of the current velocity of the center of mass with respect to the longitudinal axis of the car, $(x_i, y_i)$ is the position vector of the vehicle on a two-dimensional plane, $\theta_i$ is the heading angle, $v_i$ is the speed of the vehicle, and the two control inputs $a_i$ and $\alpha_i$ are respectively the linear acceleration and steering angle of the vehicle.
The vehicle's dynamics are discretized with a sampling time $\Delta T > 0$, leading to the following discrete-time form
\begin{equation}
\begin{alignedat}{2}
  \label{eq:ex-disc-dyn}
  {x}_{i,k+1} &= x_{i,k} + \Delta x_{i,k}, \quad
  && {y}_{i,k+1} = y_{i,k} + \Delta y_{i,k}, \\
  {\theta}_{i,k+1} &= \theta_{i,k} + \Delta \theta_{i,k}, \quad
  && v_{i,k+1} = v_{i,k} + \Delta T a_{i,k}. \\
\end{alignedat}
\end{equation}
in which the one-step changes $\Delta x_{i,k}$, $\Delta y_{i,k}$ and $\Delta \theta_{i,k}$ are nonlinear in other variables.
In this example, these nonlinear components are learned by three GP models, 
$\Delta x_{i,k} \sim \GPModel_{i, \Delta x} ( \mathbf{x}_{p,i,k} )$, $\Delta y_{i,k} \sim  \GPModel_{i, \Delta y} ( \mathbf{x}_{p,i,k} )$, and $\Delta \theta_{i,k} \sim  \GPModel_{i, \Delta \theta} ( \mathbf{x}_{a,i,k} )$ with the vectors of GP inputs $\mathbf{x}_{p,i,k} = [\cos \theta_{i,k}, \sin \theta_{i,k}, v_{i,k}, \alpha_{i,k}]^T$ and $\mathbf{x}_{a,i,k} = [v_{i,k}, \alpha_{i,k}]^T$.
Note that the GP input vectors written in bold are different from the vehicle's position $x_{i,k}$. %
The GP models result in the following GP dynamical equations
\begin{equation}
  \begin{split}
    \label{eq:ex-gp-dyn}
    \Delta \bar{x}_{i,k} &= \mu_{i,\Delta x} ( \textbf{x}_{p,i,k} ), \quad
    \Delta \bar{y}_{i,k} = \mu_{i, \Delta y} ( \textbf{x}_{p,i,k} ), \\
    \Delta \bar{\theta}_{i,k} &= \mu_{i, \Delta \theta} ( \textbf{x}_{a,i,k} ) \text.
  \end{split}
\end{equation}

The GP-DMPC formulation of this example is given by
\vspace{-5pt}
\begin{subequations}
\label{eq:ex-gpmpc}
\begin{align}
  &
  \underset{\{a_{i,k}, \alpha_{i,k}\}}{\text{minimize}} \sum_{i=1}^{M} J_i - \gamma \big( \HHH_{i,x} + \HHH_{i,y} + \HHH_{i,\theta} \big) + J_{\NNN_i} \label{eq:ex-gpmpc:obj}\\
  & \text{subject to} \nonumber \\
  & \qquad \text{\eqref{eq:ex-disc-dyn} and \eqref{eq:ex-gp-dyn}} \\
  & \qquad v_{\text{min}} \leq v_{i,k} \leq v_{\text{max}}, \label{eq:ex-gpmpc:velo-bound-cons}\\   
  & \qquad a_{\text{min}} \leq a_{i,k} \leq a_{\text{max}}, \quad \alpha_{\text{min}} \leq \alpha_{i,k} \leq \alpha_{\text{max}} \label{eq:ex-gpmpc:bound-cons} \\
  & \qquad x_{\text{min}} \leq x_{i,k} \leq x_{\text{max}}, \quad y_{\text{min}} \leq y_{i,k} \leq y_{\text{max}} \label{eq:ex-gpmpc:pos-bound}
\end{align}
\end{subequations} 
where the constraints hold for all $i = 1,\dots,M$ and $k \in \III_{t}$, \eqref{eq:ex-gpmpc:velo-bound-cons} are velocity bound constraints, \eqref{eq:ex-gpmpc:bound-cons} are bound constraints on the control inputs, \eqref{eq:ex-gpmpc:pos-bound} are safety bound constraints on the vehicle's positions to ensure that the cars move within the experimental space.
The local control objective $J_i$ is given by
\begin{equation*}
  J_i = \sum_{k = t}^{t+H-1} \left\|
  \begin{bmatrix}
    a_{i,k} \\ \alpha_{i,k}
  \end{bmatrix}
  \right\|_{\mathbf{R}_i}^2 + 
  \left\|
  \begin{bmatrix}
    x_{i,k+1} \\ y_{i,k+1}
  \end{bmatrix} - \mathbf{r}_{i,k+1} \right\|_{\mathbf{Q}_i}^2 
\end{equation*}
where $\mathbf{r}_{k+1}$ denotes the reference at time step $k+1$.
Note that $\mathbf{Q}_i \neq \mathbf{0}$ for the lead vehicle and $\mathbf{Q}_i = \mathbf{0}$ for the other vehicles, \ie only the lead vehicle is required to track a reference.
Note that given a vector $\nu$ and a positive semidefinite matrix $\mathbf{M}$, we define $\norm{\nu}_{\mathbf{M}}^2 = \nu^T \mathbf{M} \nu$.  
The active learning goals for the GP models $\HHH_{i,x}$, $\HHH_{i,y}$, and $\HHH_{i,\theta}$ are given by
\begin{equation*}
\begin{split}
\HHH_{i,x} &= \operatorname{{log\,det}} \big(\mathbf{\Sigma}_{i,\GPModel_{\Delta x}} (\mathbf{x}_{p,i,t+1:t+H}) \big), \\
\HHH_{i,y} &= \operatorname{{log\,det}} \big(\mathbf{\Sigma}_{i,\GPModel_{\Delta y}} (\mathbf{x}_{p,i,t+1:t+H}) \big), \\
\HHH_{i,\theta} &= \operatorname{{log\,det}} \big(\mathbf{\Sigma}_{i,\GPModel_{\Delta \theta}} (\mathbf{x}_{a,i,t+1:t+H}) \big) . 
\end{split}
\end{equation*}
where $\mathbf{x}_{p,i,t+1:t+H}$ and $\mathbf{x}_{a,i,t+1:t+H}$ denote the concatenated vector of GP inputs from time $t+1$ to $t+H$.
The shared objective function $J_{\NNN_i}$ encodes the formation goal as
\begin{equation*}
  J_{\NNN_i} = \hspace{-5pt} \sum_{j \in \NNN_i, j \neq i} \sum_{k = t}^{t+H-1} \left\|
  \begin{bmatrix}
    x_{i,k+1} \\ y_{i,k+1}
  \end{bmatrix} -
  \begin{bmatrix}
    x_{j,k+1} \\ y_{j,k+1}
  \end{bmatrix} -
  \Delta_{i,j}  \right\|_{\mathbf{P}_{i,j}}^2
\end{equation*}
in which $\Delta_{i,j}$ is the predefined distance between vehicles $i$ and $j$ in the formation.

We assume that the agents have different unknown system parameters, therefore their GP dynamical models are different and must be learned separately.
To save time and effort spent on experiments for training data collection, we aim to conduct one experiment where all agents collect online data for model learning in a simultaneous manner.
Consequently, in the experiment design problem, 
the agents are required to perform the active learning while ensuring a predefined formation for collision avoidance and connectivity maintenance. 
We assume that three initial GP models with 100 data points for each are available, for example, the models learned from historical data of one particular vehicle. 
These models are used as the universal starting models for all agents, then the distributed experiment design method is applied in 100 time steps to collect new data points and retrain the individual GP models for each agent, while the older data points are sequentially discarded. 
Meanwhile, in the control problem, the active learning objectives are disabled, the vehicles collaborate to perform a control task where the lead vehicle tracks a reference while the entire network form and remain a formation.

The sampling time $\Delta T$ was chosen to be \SI{200}{ms} while the control horizon length was 5.
The system parameters of the vehicles were chosen by random perturbation up to $20\%$ of the following nominal values: 
$l_r = \SI{0.386}{m}$,
$l_f = \SI{0.205}{m}$.
The constant parameters in the control problem \eqref{eq:ex-gpmpc} were, for all $i$:
$v_{\mathrm{min}} = \SI{0}{m/s}$,
$v_{\mathrm{max}} = \SI{2}{m/s}$,
$a_{\mathrm{min/max}} = \SI{\pm 2}{m/s^2}$,
$\alpha_{\mathrm{min/max}} = \SI{\pm \pi/4}{rad}$,
$x_{\mathrm{min/max}} = y_{\mathrm{min/max}} = \SI{\pm 10}{m}$,
$P_{i,i+1} = \operatorname{diag}([10,10])$,
$Q_l = \operatorname{diag}([10^2,10^2])$,
$R_i = \operatorname{diag}([0.1, 0.1])$,
$\gamma = 10$.
The parameters of the ADMM-C algorithm were:
$k_{\mathrm{max}} = 10$,
$\rho = 10^2$,
$r_i^{(0)} = 0.1$,
$\beta_{\mathrm{fail}} = 0.5$,
$\beta_{\mathrm{succ}} = 2.0$,
$\epsilon_{0} = 0.2$,
$\epsilon_{1} = 0.4$,
$\epsilon_{2} = 0.8$.

\subsection{Results and Discussions}

We conducted three simulations for the networks of 5, 9, and 15 vehicles.
The trajectories of all the vehicles in the 5-vehicles simulation case are given in Figures~\ref{fig:oed} and ~\ref{fig:control}, for the experiment and the coordination, respectively.
In both scenarios, the vehicle network is required to maintain a predefined formation, while the lead vehicle (in the middle) additionally track a figure-eight reference trajectory in the coordination task.
At the beginning of the experiment, the formation is not formed well since the GP models are not sufficiently accurate.
However, as the active learning objectives drive the agents to the states associated with new informative data, the precision of the learned GP models is gradually improved, thus the agents can maintain the formation better.
Using the models obtained from the experiment, the network of vehicles is able to perform tracking and formation control in the coordination simulation as shown in Figure~\ref{fig:control}.
The tracking errors in $x$ and $y$ positions of the lead vehicle in the tracking control task are shown in Figure~\ref{fig:tracking-err}. 
Though the lead vehicle does not perfectly track the reference, the tracking errors are kept small within $\SI{0.15}{m}$ during the steady state.
The simulation results for 9-vehicles and 15-vehicles simulations are available at the video \url{https://youtu.be/U9bunkfFqnE}. 

\begin{figure}[!tb]
    \vspace{10pt}
    \centering
    \scalebox{0.75}{\input{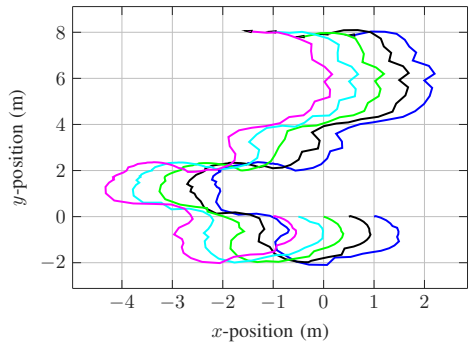}}
    \caption{Trajectories of the vehicles in the experiment.}
    \label{fig:oed}
    \vspace{-5pt}
\end{figure}

\begin{figure}[!tb]
    \centering
    \scalebox{0.75}{\input{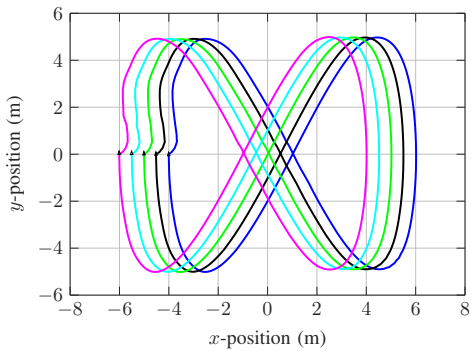}}
    \caption{Trajectories of the vehicles in the coordination.}
    \label{fig:control}
\end{figure}

Figure~\ref{fig:timing} shows the statistical boxplots of the solving time per time step of the ADMM-C algorithm in the experiment and the coordination, in three simulation cases with 5, 9 and 15 vehicles, respectively.
All simulations in this work are performed on a DELL computer with a 3.0 GHz Intel Core i5 CPU and 8 Gb RAM, while the Julia programming language is used for the implementation.
Overall, as the number of vehicles increases, the algorithm takes increasingly longer time to solve the problem.
Additionally, it can be seen that the computation time required for solving the experiment design problem which involves the log determinant of the GP covariance matrix is not much higher than that for solving the coordination problem.
Note that the computation time also scales proportionally with the predefined number of iterations in the algorithm which is chosen appropriately to balance the control performance and the computational practicality. 

\begin{figure}[!tb]
    \vspace{7pt}
    \centering
    \scalebox{0.75}{\input{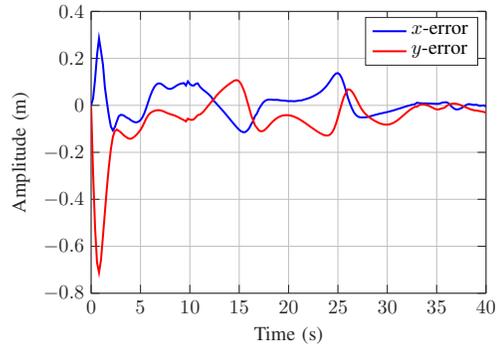}}
    \caption{Tracking errors in $x$ and $y$ positions of the lead vehicle.}
    \label{fig:tracking-err}
\end{figure}

\begin{figure}[!tb]
    \centering
    \scalebox{0.75}{\input{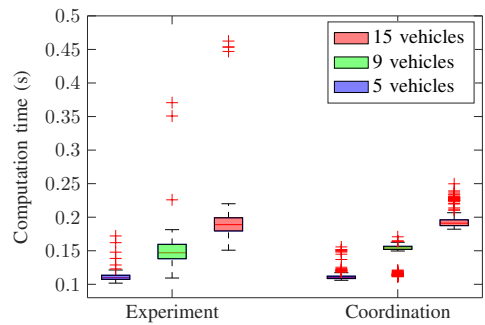}}
    \caption{Boxplots of computation time of the ADMM-C algorithm in the experiment and the coordination.} \label{fig:timing}
\end{figure}

\section{Conclusion}
\label{sec:conclusion}

We presented a Gaussian Process-based Distributed Model Predictive Control problem for multi-agent systems that covers the experiment design and coordination problems. 
The ADMM with Convexification (ADMM-C) optimization algorithm was developed to solve the resulting non-convex and complex problem in a distributed manner, in which the first-order approximations of the active learning objectives and the mean of GP dynamics were utilized to convexify the local subproblem at each itearation.
Under some technical assumptions, we proved that the proposed algorithm converges to a stationary point of the penalized optimization problem.
The performance of our problem formulation and distributed optimization method was validated by a numerical simulation of a multi-vehicle formation control system.
Our future work aims to improve the performance and scalability of our approach and then apply it in real-world systems.

\bibliographystyle{IEEEtran}
\bibliography{IEEEabrv,references}

\appendices
\section{Convergence properties of ADMM-C algorithm} 
\label{sec:convergence}

The convergence analysis of the ADMM-C algorithm is given in this section.
To prove that the ADMM-C converges to a stationary point of the problem \eqref{eq:admm-prob}, we will show that the following four key properties \cite{wang2019global} are satisfied:

\begin{itemize}
\item P1 (Boundedness) $\{ \mathbf{x}_i^{(k)}, \mathbf{z}_{\NNN_i}^{(k)}, \mathbf{y}_{\NNN_i}^{(k)}\}_{i \in \VVV}$ is bounded, and $L_\rho ( \{ \mathbf{x}_{i}^{(k)}, \mathbf{z}_{\NNN_i}^{(k)}, \mathbf{y}_{\NNN_i}^{(k)} \}_{i \in \VVV} )$ is lower bounded.

\item P2 (Sufficient descent) There is a constant $c_1 > 0$ such that for all $k \in \NN$ and $i \in \VVV$ we have
\begin{equation}
\begin{multlined}
\hspace{-5pt} L_\rho \big( \{ \mathbf{x}_{i}^{(k)}, \mathbf{z}_{\NNN_i}^{(k)}, \mathbf{y}_{\NNN_i}^{(k)} \}_{i \in \VVV} \big) \\
- L_\rho \big( \{ \mathbf{x}_{i}^{(k+1)}, \mathbf{z}_{\NNN_i}^{(k+1)}, \mathbf{y}_{\NNN_i}^{(k+1)} \}_{i \in \VVV} \big) \\
\hspace{-15pt} \ge \hspace{-1pt} c_1 \hspace{-1pt} \sum_{i \in \VVV} \hspace{-2pt} \Big( \norm{\mathbf{x}_i^{(k+1)} - \mathbf{x}_i^{(k)}}^2 \hspace{-2pt}+\hspace{-2pt} \norm{\mathbf{z}_{\NNN_i}^{(k+1)} - \mathbf{z}_{\NNN_i}^{(k)}}^2 \Big)
\end{multlined}
\end{equation}

\item P3 (Subgradient bound) There is a constant $c_2 > 0$ and $d^{(k+1)} \in \partial L_\rho ( \{ \mathbf{x}_{i}^{(k+1)}, \mathbf{z}_{\NNN_i}^{(k+1)}, \mathbf{y}_{\NNN_i}^{(k+1)} \}_{i \in \VVV} )$ such that
\begin{equation}
\begin{multlined}
\norm{d^{(k+1)}} \le c_2 \sum_{i \in \VVV} \Big( \norm{\mathbf{x}_i^{(k+1)} - \mathbf{x}_i^{(k)}} \\
+ \norm{\mathbf{z}_{\NNN_i}^{(k+1)} - \mathbf{z}_{\NNN_i}^{(k)}} \Big)
\end{multlined}
\end{equation}

\item P4 (Limiting continuity) If $\{ \mathbf{x}_i^{*}, \mathbf{z}_{\NNN_i}^{*}, \mathbf{y}_{\NNN_i}^{*}\}_{i \in \VVV}$ is the limit point of a sub-sequence $\{ \mathbf{x}_i^{(k_s)}, \mathbf{z}_{\NNN_i}^{(k_s)}, \mathbf{y}_{\NNN_i}^{(k_s)}\}_{i \in \VVV}$, then 
\[
\begin{multlined}
L_\rho \big( \{ \mathbf{x}_i^{*}, \mathbf{z}_{\NNN_i}^{*}, \mathbf{y}_{\NNN_i}^{*}\}_{i \in \VVV} \big) \\
= \lim_{s \to \infty} L_\rho \big( \{ \mathbf{x}_i^{(k_s)}, \mathbf{z}_{\NNN_i}^{(k_s)}, \mathbf{y}_{\NNN_i}^{(k_s)}\}_{i \in \VVV} \big)
\end{multlined}
\]

\end{itemize}

The above properties and their proofs will be given in Lemmas~\ref{L6-admm}, \ref{L9-admm}, \ref{L10-admm}, and \ref{LP4-admm}, respectively. 

\begin{remark}
Given the Assumption~\ref{A5-admm}, for all $i \in \VVV$, $f_{\NNN_i}$ are Lipschitz differentiable and assume that $\bar{L}$ is the universal Lipschitz constant for all $\nabla f_{\NNN_i}$.
\end{remark}

\begin{remark}
For any fixed $\mathbf{u}$ with appropriate dimension, $\underset{\mathbf{z}_{\NNN_i}} {\operatorname{minimize}}\; \{ f_{\NNN_i} (\mathbf{z}_{\NNN_i}) : \mathbf{z}_{\NNN_i} = \mathbf{u} \}$ always has a unique minimizer and $H_i(\mathbf{u}) \triangleq \underset{\mathbf{z}_{\NNN_i}} {\operatorname{argmin}}\; \{ f_{\NNN_i} (\mathbf{z}_{\NNN_i}) : \mathbf{z}_{\NNN_i} = \mathbf{u} \} = \mathbf{u}$ is a Lipschitz continuous map with Lipschitz constant $\bar{M} = 1$. 
Therefore, the Assumption A3(a) in \cite{wang2019global} is satisfied for our problem \eqref{eq:admm-prob}.
\end{remark}

\begin{remark}
Given the Assumption~\ref{A-coercivity}, the Assumption A1 in \cite{wang2019global} on coercivity is satisfied for our problem \eqref{eq:admm-prob}. 
\end{remark}

\begin{remark}
Given the Assumption~\ref{A-feasibility}, the Assumption A2 in \cite{wang2019global} on feasibility is satisfied for our problem \eqref{eq:admm-prob}.
\end{remark}

\begin{lemma}
\label{L3.17-scp}
For any iteration $k \in \NN$, there exists a constant $a > 0$ such that the accepted solution of the problem \eqref{eq:admm-x} by the SCP algorithm satisfy
\begin{equation}
\label{eq:sufficient-descent-scp}
\begin{multlined}
\bar{L}_{\rho, i} ( \mathbf{x}_i^{(k)}, \mathbf{z}_{\NNN_i}^{(k)}, \mathbf{y}_{\NNN_i}^{(k)} ) - \bar{L}_{\rho, i} ( \mathbf{x}_i^{(k+1)}, \mathbf{z}_{\NNN_i}^{(k)}, \mathbf{y}_{\NNN_i}^{(k)} ) \\
\ge a \norm{\mathbf{x}_i^{(k+1)} - \mathbf{x}_i^{(k)}}^2, \forall i \in \VVV
\end{multlined}
\end{equation}
\end{lemma}

The proof follows directly from the proof of Condition 3.17 in \cite{mao2018successive}.

\begin{lemma}
\label{L3.18-scp}
Given Assumptions~\ref{A-convexity} and \ref{A3.15-scp}, for any iteration $k \in \NN$, there exists a constant $b > 0$ and $d_i^{k+1} \in \frac{\partial \bar{L}_{\rho,i}}{\partial \mathbf{x}_i} \big( \mathbf{x}_{i}^{(k+1)}, \mathbf{z}_{\NNN_i}^{(k+1)}, \mathbf{y}_{\NNN_i}^{(k+1)} \big)$ for all $i \in \VVV$ such that
\begin{equation}
\norm{d_i^{k+1}} \le b \norm{\mathbf{x}_i^{(k+1)} - \mathbf{x}_i^{(k)}}, \forall i \in \VVV
\end{equation}
\end{lemma}

The proof follows directly from the proof of Condition 3.18 in \cite{mao2018successive}.

\begin{lemma}
\label{L5-admm}
If $\rho > 4\bar{L} + 2$, then for any iteration $k \in \NN$
\begin{equation}
\label{eq:sufficient-descent-admm} 
\begin{multlined}
L_{\rho, i} ( \mathbf{x}_{\NNN_i}^{(k+1)}, \mathbf{z}_{\NNN_i}^{(k)}, \mathbf{y}_{\NNN_i}^{(k)} ) \\ 
\hspace{-5pt} - L_{\rho, i} ( \mathbf{x}_{\NNN_i}^{(k+1)}, \mathbf{z}_{\NNN_i}^{(k+1)}, \mathbf{y}_{i}^{(k+1)} ) 
\ge \norm{\mathbf{z}_{\NNN_i}^{(k+1)} - \mathbf{z}_{\NNN_i}^{(k)}}^2
\end{multlined}
\end{equation}
\end{lemma}

The proof follows directly from the proof of Lemma 5 in \cite{wang2019global}. 
Note that for our problem $\bar{M} = 1$.

\begin{lemma}
\label{L6-admm}
Given the Assumption~\ref{A-coercivity} and if $\rho >  4\bar{L} + 2$, then the sequence $\{ \mathbf{x}_i^{(k)}, \mathbf{z}_{\NNN_i}^{(k)}, \mathbf{y}_{\NNN_i}^{(k)}\}_{i \in \VVV}$ generated by Algorithm~\ref{alg:ADMM-C} satisfies
\begin{enumerate}
\item 
$ \begin{multlined} 
L_\rho \big( \{ \mathbf{x}_{i}^{(k)}, \mathbf{z}_{\NNN_i}^{(k)}, \mathbf{y}_{\NNN_i}^{(k)} \}_{i \in \VVV} \big) \\
\ge L_\rho \big( \{ \mathbf{x}_{i}^{(k+1)}, \mathbf{z}_{\NNN_i}^{(k+1)}, \mathbf{y}_{\NNN_i}^{(k+1)} \}_{i \in \VVV} \big)
\end{multlined} $.
\item $L_\rho \big( \{ \mathbf{x}_{i}^{(k)}, \mathbf{z}_{\NNN_i}^{(k)}, \mathbf{y}_{\NNN_i}^{(k)} \}_{i \in \VVV} \big) $ is lower bounded for all k and converges as $k \to \infty$.
\item $\{ \mathbf{x}_i^{(k)}, \mathbf{z}_{\NNN_i}^{(k)}, \mathbf{y}_{\NNN_i}^{(k)}\}_{i \in \VVV}$ is bounded
\end{enumerate}
\end{lemma}

\begin{proof}
Part 1. 
Take the sum of all inequalities in Lemma~\ref{L3.17-scp} and from \eqref{eq:admm-lagrang2}, we have
\begin{equation}
\label{eq:descent1}
\begin{multlined}
L_\rho \big( \{ \mathbf{x}_{i}^{(k)}, \mathbf{z}_{\NNN_i}^{(k)}, \mathbf{y}_{\NNN_i}^{(k)} \}_{i \in \VVV} \big) \\
- L_\rho \big( \{ \mathbf{x}_{i}^{(k+1)}, \mathbf{z}_{\NNN_i}^{(k)}, \mathbf{y}_{\NNN_i}^{(k)} \}_{i \in \VVV} \big) \\
\ge a \sum_{i \in \VVV} \norm{\mathbf{x}_i^{(k+1)} - \mathbf{x}_i^{(k)}}^2
\ge 0
\end{multlined}
\end{equation}

Take the sum of all inequalities in Lemma~\ref{L5-admm} and from \eqref{eq:admm-lagrang1}, we have
\begin{equation}
\label{eq:descent2}
\begin{multlined}
L_\rho \big( \{ \mathbf{x}_{i}^{(k+1)}, \mathbf{z}_{\NNN_i}^{(k)}, \mathbf{y}_{\NNN_i}^{(k)} \}_{i \in \VVV} \big) \\
- L_\rho \big( \{ \mathbf{x}_{i}^{(k+1)}, \mathbf{z}_{\NNN_i}^{(k+1)}, \mathbf{y}_{\NNN_i}^{(k+1)} \}_{i \in \VVV} \big) \\
\ge \sum_{i \in \VVV} \norm{\mathbf{z}_{\NNN_i}^{(k+1)} - \mathbf{z}_{\NNN_i}^{(k)}}^2
\ge 0
\end{multlined}
\end{equation}

As a result, we obtain
$L_\rho \big( \{ \mathbf{x}_{i}^{(k)}, \mathbf{z}_{\NNN_i}^{(k)}, \mathbf{y}_{\NNN_i}^{(k)} \}_{i \in \VVV} \big) 
\ge L_\rho \big( \{ \mathbf{x}_{i}^{(k+1)}, \mathbf{z}_{\NNN_i}^{(k+1)}, \mathbf{y}_{\NNN_i}^{(k+1)} \}_{i \in \VVV} \big)$. 

Part 2. Follows the proof of Lemma 6, part 2 in \cite{wang2019global}

Part 3. Follows the proof of Lemma 6, part 3 in \cite{wang2019global}. 
\end{proof}

From part 2 and part 3 of Lemma~\ref{L6-admm}, the boundeness property P1 holds.

\begin{lemma}
\label{L9-admm}
If $\rho >  4\bar{L} + 2$ then Algorithm~\ref{alg:ADMM-C} satisfies the sufficient descent property P2.
\end{lemma}

\begin{proof}
From \eqref{eq:descent1} and \eqref{eq:descent2}, we have
\begin{equation}
\begin{split}
& 
\begin{multlined}
L_\rho \big( \{ \mathbf{x}_{i}^{(k)}, \mathbf{z}_{\NNN_i}^{(k)}, \mathbf{y}_{\NNN_i}^{(k)} \}_{i \in \VVV} \big) \\
- L_\rho \big( \{ \mathbf{x}_{i}^{(k+1)}, \mathbf{z}_{\NNN_i}^{(k+1)}, \mathbf{y}_{\NNN_i}^{(k+1)} \}_{i \in \VVV} \big)
\end{multlined} \\
& \ge a \sum_{i \in \VVV} \norm{\mathbf{x}_i^{(k+1)} - \mathbf{x}_i^{(k)}}^2 +\sum_{i \in \VVV} \norm{\mathbf{z}_{\NNN_i}^{(k+1)} - \mathbf{z}_{\NNN_i}^{(k)}}^2 \\
& 
\begin{multlined}
\ge c_1 \sum_{i \in \VVV} \Big ( \norm{\mathbf{x}_i^{(k+1)} - \mathbf{x}_i^{(k)}}^2 %
+ \norm{\mathbf{z}_{\NNN_i}^{(k+1)} - \mathbf{z}_{\NNN_i}^{(k)}}^2 \Big) 
\end{multlined} 
\end{split}
\end{equation} 
where $c_1 = \operatorname{min} (a, 1)$.
Therefore the sufficient descent property P2 holds. \end{proof}

\begin{lemma}
\label{L10-admm}
Algorithm~\ref{alg:ADMM-C} satisfies the subgradient bound property P3.
\end{lemma}

\begin{proof}
We have 
\begin{equation}
\begin{multlined}
\partial L_\rho \big( \{ \mathbf{x}_{i}^{(k+1)}, \mathbf{z}_{\NNN_i}^{(k+1)}, \mathbf{y}_{\NNN_i}^{(k+1)} \}_{i \in \VVV} \big) \\ 
= \left( \left\{ \frac{\partial L_\rho}{\mathbf{x}_{i}} \right\}_{i \in \VVV}, \{ \nabla_{\mathbf{z}_{\NNN_i}} L_\rho \}_{i \in \VVV}, \{\nabla_{\mathbf{y}_{\NNN_i}} L_\rho \}_{i \in \VVV} \right)
\end{multlined}
\end{equation}

Hence follows the proof of Lemma 10 in \cite{wang2019global}, we need to show for each $i \in \VVV$, there exists a constant $b > 0$ and $d_i^{k+1} \in \frac{\partial \bar{L}_{\rho,i}}{\partial \mathbf{x}_i} \big( \mathbf{x}_{i}^{(k+1)}, \mathbf{z}_{\NNN_i}^{(k+1)}, \mathbf{y}_{\NNN_i}^{(k+1)} \big)$ 
such that
\begin{equation}
\label{eq:P3-cond1}
\norm{d_i^{k+1}} \le b \norm{\mathbf{x}_i^{(k+1)} - \mathbf{x}_i^{(k)}}^2
\end{equation}
and 
\begin{equation}
\label{eq:P3-cond2}
\begin{multlined}
\nabla_{\mathbf{y}_{\NNN_i}} L_{\rho,i} (\mathbf{x}_{i}^{(k+1)}, \mathbf{z}_{\NNN_i}^{(k+1)}, \mathbf{y}_{\NNN_i}^{(k+1)}) \\
\le \frac{\bar{L}}{\rho} \norm{\mathbf{z}_{\NNN_i}^{(k+1)} - \mathbf{z}_{\NNN}^{(k)}} 
\end{multlined}
\end{equation}
\begin{equation}
\label{eq:P3-cond3}
\begin{multlined}
\nabla_{\mathbf{z}_{\NNN_i}} L_{\rho,i} (\mathbf{x}_{i}^{(k+1)}, \mathbf{z}_{\NNN_i}^{(k+1)}, \mathbf{y}_{\NNN_i}^{(k+1)}) \\
\le \bar{L} \norm{\mathbf{z}_{\NNN_i}^{(k+1)} - \mathbf{z}_{\NNN_i}^{(k)}} 
\end{multlined}
\end{equation}

The inequality \eqref{eq:P3-cond1} is given by Lemma~\ref{L3.18-scp}, while \eqref{eq:P3-cond2} and \eqref{eq:P3-cond3} was proven in \cite{wang2019global} (Eq. (28) and (29), respectively). 
As a result, there exists $d^{(k+1)} \in \partial L_\rho ( \{ \mathbf{x}_{i}^{(k+1)}, \mathbf{z}_{\NNN_i}^{(k+1)}, \mathbf{y}_{\NNN_i}^{(k+1)} \}_{i \in \VVV} )$ such that
\begin{equation}
\begin{multlined}
\norm{d^{(k+1)}} \le c_2 \sum_{i \in \VVV} \Big( \norm{\mathbf{x}_i^{(k+1)} - \mathbf{x}_i^{(k)}} \\
+ \norm{\mathbf{z}_{\NNN_i}^{(k+1)} - \mathbf{z}_{\NNN_i}^{(k)}} \Big)
\end{multlined}
\end{equation}
where $c_2 = \operatorname{max} (b, \bar{L}/\rho, \bar{L})$.
The proof is therefore completed. \end{proof}

\begin{lemma}
\label{LP4-admm}
Algorithm~\ref{alg:ADMM-C} satisfies the limiting continuity property P4.
\end{lemma}

\begin{proof}
Since $L_\rho$ is continuous due to Assumptions~\ref{A5-admm} and \ref{A3.15-scp}, the proof of P4 in \cite{wang2019global} for a general case with lower semicontinuous function can be applied. \end{proof}

From Lemmas~\ref{L6-admm}, \ref{L9-admm}, \ref{L10-admm}, and \ref{LP4-admm} that guarantee the properties P1-P4, we are now able to state a theorem on the convergence property of the ADMM-C algorithm.

\begin{theorem}
\label{theo:convergence}
Suppose that Assumptions 1-5 hold.
If $\rho >  4\bar{L} + 2$ then Algorithm~\ref{alg:ADMM-C} generates a sequence that is bounded, has at least one limit point, and each limit point $\{ \mathbf{x}_i^{*}, \mathbf{z}_{\NNN_i}^{*}, \mathbf{y}_{\NNN_i}^{*}\}_{i \in \VVV}$ is a stationary point of $L_\rho$. 
Note that this stationary point is not globally unique unless $L_\rho$ is a Kurdyka-Lojasiewicz (KL) function \cite{attouch2013convergence}.
\end{theorem}

The proof follows directly from the proof of Proposition 2 in \cite{wang2019global}.

\end{document}